\PassOptionsToPackage{svgnames}{xcolor}
\documentclass[12pt]{article}

\usepackage{needspace,amsmath,empheq,amssymb,amsfonts,amsthm,setspace,tikz,dsfont,wrapfig,lipsum,natbib,longtable, array,colortbl,mathrsfs,pifont,wasysym,nicefrac,multicol,booktabs,makecell,multirow,thm-restate,caption, subcaption,setspace,makecell}
\usetikzlibrary{patterns}
\usepackage{etoolbox}
\usepackage{thmtools,thm-restate}
\usepackage{caption}
\usepackage[all]{nowidow}
\usepackage[makeroom]{cancel}

\def\equationautorefname~#1\null{%
	Expression~(#1)\null
}

\usepackage[letterpaper]{geometry}
\usepackage{hyperref}

\hypersetup{                                             
    pdffitwindow=true,                                  
    pdfstartview={XYZ null null 1.00},                   
    pdfnewwindow=true,                                   
    colorlinks=true,                                     
    linkcolor=DarkBlue,                                     
    citecolor=DarkBlue,                                     
    urlcolor=DarkBlue,                                      
    pdfauthor = {Gabriel Ziegler},
    pdfkeywords = {COVID-19, SARS-CoV-2, Prevalence, Austria},
    pdftitle = {Prevalence_AUT},
    pdfsubject = {Prevalence_AUT},
    pdfpagemode = UseNone}

  \usepackage{pgfplots}
  \pgfplotsset{compat=newest}
  \usetikzlibrary{plotmarks}
  \usetikzlibrary{arrows.meta}
  \usepgfplotslibrary{patchplots}
  \usepackage{grffile}
  \usetikzlibrary{positioning}

\pgfplotsset{plot coordinates/math parser=false}
\newlength\figureheight
\newlength\figurewidth
\usepgfplotslibrary{colormaps}
\usetikzlibrary{pgfplots.colormaps} 

\usepackage{fancyhdr,ifthen}
\fancypagestyle{appendix}{
\fancyhead{}
\fancyhead[R]{\ifthenelse{\isodd{\value{page}}}{\thepage}{}}
\fancyhead[L]{\ifthenelse{\isodd{\value{page}}}{\textsc{}}{\thepage}}

\fancyfoot{}
}

\newtheorem{assumption}{Assumption}

\newtheorem{proposition}{Proposition}

\theoremstyle{definition}

\usepackage{titlesec}

\titleformat*{\section}{\large\scshape\centering}
\titleformat*{\subsection}{\scshape\centering}
\titleformat*{\subsubsection}{\itshape}
\titleformat*{\paragraph}{\large\bfseries\centering}
\titleformat*{\subparagraph}{\large\bfseries\centering}

\makeatletter
\DeclareRobustCommand\citepos
  {\begingroup
   \let\NAT@nmfmt\NAT@posfmt
   \NAT@swafalse\let\NAT@ctype\z@\NAT@partrue
   \@ifstar{\NAT@fulltrue\NAT@citetp}{\NAT@fullfalse\NAT@citetp}}

\let\NAT@orig@nmfmt\NAT@nmfmt
\def\NAT@posfmt#1{\NAT@orig@nmfmt{#1's}}

\makeatother

\title{How many people are infected?\\{\large\vspace*{-0.5em}A case study on SARS-CoV-2 prevalence in Austria}\footnote{Thanks to Daniel Ladenhauf for providing me with the information about the data. All errors are mine.}}

\author{
Gabriel Ziegler\footnote{The University of Edinburgh, School of Economics; 31 Buccleuch Place, Edinburgh, EH8 9JT, UK; \href{mailto:ziegler@ed.ac.uk}{\tt ziegler@ed.ac.uk}.}
} 

\begin{document}

\maketitle
\thispagestyle{empty}
\vspace*{-2em}
\begin{abstract}
	\small
	Using recent data from voluntary mass testing, I provide credible bounds on prevalence of SARS-CoV-2 for Austrian counties in early December 2020. When estimating prevalence, a natural missing data problem arises: no test results are generated for non-tested people. In addition, tests are not perfectly predictive for the underlying infection. This is particularly relevant for mass SARS-CoV-2 testing as these are conducted with rapid Antigen tests, which are known to be somewhat imprecise. Using insights from the literature on partial identification, I propose a framework addressing both issues at once.
	I use the framework to study differing selection assumptions for the Austrian data. Whereas weak monotone selection assumptions provide limited identification power, reasonably stronger assumptions reduce the uncertainty on prevalence significantly.
	\vspace{0.2cm}
	
\end{abstract}
\cleardoublepage
\setcounter{page}{1}

\onehalfspacing
\section{Introduction}
\label{section:introduction}
An important measure for (health) policy during a pandemic is prevalence. As a quantification of current disease status, prevalence is the proportion of people in a given population having the disease. \citep{rothman-2012} Besides providing an evaluation of how widespread the diseases is, prevalence is relevant in how accurate diagnostic tests are. \citep{zhou-etal-2014} However, prevalence is not directly observable in many situations.

Often prevalence is inferred from diagnostic tests results, which are indicative of the disease. There are at least two problems arising in such a situation. First, tests are usually not perfectly indicative for the disease. A test might produce either false positives, false negatives, or both. Recently, \cite{ziegler-2020} discusses how this problem worsens when the test accuracy is evaluated with respect to an imperfect reference test. In such cases, the test's information is ambiguous. Second, the tested population is usually not the whole population and the testing pool'composition matters for inference of overall prevalence. The latter problem is even more severe in cases, when then the composition of the testing pool is unknown. For example, when testing is voluntary it is not obvious whether disease-susceptible people are more or less likely to take the test. That is, people self-select into the testing pool. Selection problems as occurring with voluntary testing are ubiquitous within Economics. Without strong assumption on unobservable data, \cite{manski-1989} shows that this problem leads to an identification problem and therefore it might not be possible to assign a unique number to the relevant statistic.

In this note, I use data of voluntary COVID-19 mass testing in Austria in December 2020 to provide bounds on SARS-CoV-2 (point) prevalance at that time. Building on the work of \cite{manski-molinari-2021}, \cite{stoye-2020}, and \cite{ziegler-2020},\footnote{\cite{sacks-etal-2020} use a related approach.} I address both of the problems mentioned within one framework. This allows me to illustrate how much knowledge about prevalence can be obtained just from the data alone (with minimal assumptions). Furthermore, the framework provides a simple method to address the identifying power of varying (stronger) assumptions about the selection problem.  

In the first half (4/12--15/12) of December 2020, every Austrian municipality provided voluntary SARS-CoV-2 tests for their population via rapid Antigen tests. In many municipalities testing was available only for a few, consecutive days. The goal of the policy was to identify otherwise undetected SARS-CoV-2 infected people. For this, people with typical symptoms, people who were tested regularly before in their workplace, quarantined people, and children below school starting age\footnote{Austrian school starting age is 6 years with September 1st as cutoff date.} were explicitly asked to \emph{not} attend the testing. \citep{sozialministerium-2020} This and anecdotal evidence suggests negative selection into testing. That is, tested people are less susceptible of being infected by SARS-CoV-2. \cite{derstandard-2020} provides data of test results and participation on the county level. The dataset only covers $7$ out of the $9$ Austrian states (Bundesl\"{a}nder).\footnote{Data on Vorarlberg is missing. Data for Carinthia does not include the number of positive tests results and is therefore omitted from the analysis. Three counties are combined due to local misreporting of data (Amstetten, Scheibbs, Waidhofen a.d. Ybbs). Since the analysis is county-by-county, it is unaffected by these data issues.}

\section{Theoretic Framework}
\label{section:theory}

Let $c=1$ denote a person who is infected with SARS-CoV-2 and $c=0$ otherwise. Participation in the mass testing is indicted with $t=1$ (again $t=0$ otherwise). Only if the person was tested, she can obtain a positive test result denoted with $a=1$ (and $a=0$ otherwise). The population (of a county) is a distribution $P(a,c,t)$, but observed data are just $P(a|t=1)$. In particular, note that $P(a=1) = P(a=1|t=1) P(t=1)$ because a positive test can only be observed for tested people.
\begin{align*}
	\gamma &:= P(a=1|t=1) &\ldots &\text{test yield} \\
	\rho &:= P(c=1) &\ldots &\text{prevalance} \\
	\tau &:= P(t=1) &\ldots &\text{proportion of tested people}.
\end{align*}

Accuracy of the test is given by sensitivity and specificity given by
\begin{align}
	\sigma &:= P(a=1|c=1,t=1) = \frac{P(a=1,c=1|t=1)}{P(c=1|t=1)} \\
	\pi &:= P(a=0|c=0,t=1) = \frac{P(a=0,c=0|t=1)}{P(c=0|t=1)},
\end{align}
respectively. In line with \cite{ziegler-2020}, Antigen tests correspond to ambiguous information and therefore I make the following assumption the both, sensitivity and specificity, are only known to lie within an interval.

\begin{assumption}[Ambiguous Information]\label{assumption:ambi_test}
	The test satisfies
	\begin{align*}
		\sigma \in \left[\underline{\sigma}, \overline{\sigma}\right] 
		\quad \text{and} \quad 
		\pi \in \left[\underline{\pi}, \overline{\pi}\right].
	\end{align*}
\end{assumption}

\autoref{assumption:ambi_test} alone provides sharp bounds on prevalence $\rho:=P(c=1)$:
\begin{proposition}\label{prop:no_assumption}
	If \autoref{assumption:ambi_test} holds, then
	\begin{align*}
		\rho = \left[\tau \frac{\gamma+\underline{\pi} -1}{\overline{\sigma}+\underline{\pi}-1}, \tau \frac{\gamma+\overline{\pi} -1}{\underline{\sigma}+\overline{\pi}-1} + (1-\tau)\right]
	\end{align*}
\end{proposition}

\begin{proof}
	First, consider fixed $\sigma$ and $\pi$. Then by the law of total probability
	\begin{align*}
		\gamma &= \sigma P(c=1|t=1) + (1-\pi) (1-P(c=1|t=1))  \\
		\iff\quad & P(c=1|t=1) = \frac{\gamma+\pi -1}{\sigma+\pi-1}
	\end{align*} 
	and $P(c=1|t=0) \in [0,1]$. Then, again by the law of total probability
	\begin{align*}
		\rho \in \left[\tau \frac{\gamma+\pi -1}{\sigma+\pi-1}, \tau \frac{\gamma+\pi -1}{\sigma+\pi-1} + (1-\tau)\right].
	\end{align*}
The fraction is increasing in $\pi$ and decreasing in $\sigma$. The result follows by evaluating at the respective extremes.
\end{proof}

The prevalance bounds in \autoref{prop:no_assumption} are pretty wide in applications as will be seen later. However, they are not completely trivial in the sense of just stating prevalence is bounded by $0$ and $1$ although they rely on minimal assumptions about the (untested) population. As can be seen in the proof of \autoref{prop:no_assumption}, $P(c=1|t=0)$ is trivially bounded without stronger assumptions, which leads to wide bounds on prevalence. As explained above, there is some indication of negative selection into the testing pool in the case of Austrian mass testing. This knowledge can be used to narrow bounds on prevalence. This extraneous information is formalized in \autoref{assumption:selection}.\footnote{A potentially more satisfying way of modeling selection is bounding the odds ratio. \cite{stoye-2020} uses such bounds in his analysis under the assumption of $\underline{\pi} = \overline{\pi}=1$, i.e. the test does not produce false-positives. Without this assumption, bounds on the odds ratio seem rather intractable. Furthermore, such an assumption is problematic in the application to Antigen tests.}
\begin{assumption}[Selection]\label{assumption:selection}
	The population satisfies
	\begin{align*}
		\frac{P(c=1|t=0)}{P(c=1|t=1)} \in \left[\underline{\kappa}, \overline{\kappa}\right],
	\end{align*}
	with $\underline{\kappa} \geq 0$.
\end{assumption}
When $\underline{\kappa} \geq 1$, then $P(c=1|t=0) \geq P(c=1|t=1)$ so that tested people are \emph{less} likely to be infected than untested people. This corresponds to the negative selection explained above. On the other hand, if $\overline{\kappa}\leq 1$, then there is positive selection, which seems more appropriate in the case of PCR testing. Indeed, \cite{manski-molinari-2021} use such an assumption in their study on prevalence of SARS-CoV-2. Their assumption (called test-monotonicity) corresponds to $(\underline{\kappa},\overline{\kappa}) = (0,1)$.

\begin{proposition}\label{prop:mono}
	Suppose \autoref{assumption:ambi_test} and \autoref{assumption:selection} hold. If $\overline{\kappa} \leq \frac{\sigma+\pi-1}{\gamma+\pi -1}$, then
	\begin{align*}
		\rho = \left[(\tau+(1-\tau)\underline{\kappa})\frac{\gamma+\underline{\pi} -1}{\overline{\sigma}+\underline{\pi}-1},  (\tau+(1-\tau)\overline{\kappa})  \frac{\gamma+\overline{\pi} -1}{\underline{\sigma}+\overline{\pi}-1}\right].
	\end{align*}
	Otherwise the upper bound is given by \autoref{prop:no_assumption}.
\end{proposition}

\begin{proof}
	As in the proof of \autoref{prop:no_assumption}, $P(c=1|t=1) = \frac{\gamma+\pi -1}{\sigma+\pi-1}$, but now
	$P(c=1|t=0) \in [\underline{\kappa}P(c=1|t=1) , \overline{\kappa}P(c=1|t=1)]$, which is below one because of $\overline{\kappa} \leq \frac{\sigma+\pi-1}{\gamma+\pi -1}$. Then
	\begin{align*}
		\rho \in \left[(\tau+(1-\tau)\underline{\kappa}) \frac{\gamma+\pi -1}{\sigma+\pi-1}, (\tau+(1-\tau)\overline{\kappa}) \frac{\gamma+\pi -1}{\sigma+\pi-1} \right].
	\end{align*}
	The fraction is increasing in $\pi$ and decreasing in $\sigma$. The result follows by evaluating at the respective extremes.
\end{proof}

\section{Empirical Analysis}
\label{section:empirical}
The dataset was already explained in \autoref{section:introduction}. It remains to get data on the test's accuracy as formalized in \autoref{assumption:ambi_test}. Rapid Antigen test were used in the mass testing in Austria. To best of my knowledge, there is no publicaly available data on which specific test was used by each municipality. However, I personally obtained the data for a few municipalities in Graz-Umgebung. In these municipalities the Standard Q COVID-19 Rapid Antigen Test of SD Biosensor/Roche for detection of SARS-CoV-2 was used. I will use this test \emph{as if} it was used in all the municipalities in the dataset.\footnote{Many Antigen tests currently on the market have very similar quality in terms of observed sensitivity and specificity relative to a PCR test. Therefore, the use of a different test would not change the results significantly.}

\cite{kaiser-etal-2020} provide an independent analysis of the Standard Q Antigen test. Relative to a PCR test, they find a (point estimate for) sensitivity of $89,0\%$ and specificity of $99,7\%$. Although a PCR test is (close to) perfect when it comes to specificity, it is known that it might not have perfect sensitivity. Several PCR tests were evaluated for sensitivity by \cite{alcoba-florez-etal-2020}. They find point estimates ranging from $60.2\%$ to $97.9\%$.\footnote{For all considered PCR tests, their $95\%$ confidence intervals exclude perfect sensitivity.}. Due to the imperfectness of the reference PCR test, the Antigen test does not have a unique value for sensitivity and specificity as discussed by \cite{ziegler-2020}. Using the method proposed by \cite{ziegler-2020} and evaluating across all possible PCR's sensitivities, Standard Q's accuracy is bounded by $\sigma \in [53.58\%, 87.65\%]$ and $\pi \in [99.53\%, 100\%]$.\footnote{In these calculations, I use the point estimates of \cite{kaiser-etal-2020}.}

It remains to specify the selection parameters $(\underline{\kappa},\overline{\kappa})$. For this, I will consider several cases corresponding to different assumptions about selection and explain the effects using the city of Graz (Graz-Stadt) as an example. Graz had a particpation rate of slightly more than $20\%$ and of these $0.9\%$ obtained a positive test result. The results for each county (together with participation $\tau$ and test yield $\gamma$) are shown in \autoref{tab:1} and \autoref{tab:2}.
\begin{description}
	\item[No Assumption.] This case corresponds to \autoref{prop:no_assumption}. Here any kind of selection, negative or positive, is allowed for. Correspondingly, the bounds on prevalence are rather wide. For Graz the bounds are  [$0.10\%$, $79.70\%$], which excludes the possibility of $0\%$ prevalence in contrast to many other counties in \autoref{tab:1}.
	\item[No selection.] Next, consider a scenario of no selection, which embodies a very strong assumption and one which might not be appropriate in the current context. No selection means that the testing decision was as if randomly assigned and therefore the results from the tested people is representative for the entire population. Mathematically this means ${P(c=1|t=0)}={P(c=1|t=1)}$ (or equivalently in the current framework, $\underline{\kappa}= \overline{\kappa}=1$). With this strong (and most likely unwarranted assumption), the bounds for Graz reduce to  $[0.49\%, 1.68\%]$, which still leaves a interval width of more than $1\%$. This uncertainty stems from the imprecise testing technology.
	\item[Negative Selection.] As explained before this assumption is credible in the current context, but it is still very weak as it just imposes ${P(c=1|t=0)}\geq{P(c=1|t=1)}$ (with $\underline{\kappa}=1$ and $\overline{\kappa}=\infty$). However, with negative selection only it is still possible that \emph{every} untested person is infected and therefore the upper bound is not reduced relative to the No Assumption scenario. For Graz the bounds are [$0.49\%$, $79.70\%$]. Note that the lower bound is almost five times higher than in the No Assumption scenario, meaning that this assumption---although weak---has quite some power for the 'best-case' prevalence.
	\item[Restricted Negative Selection.] Here, the selection assumption from before is maintained, but there is also an upper bound on selection. In particular, at most  ${P(c=1|t=0)}=2\times{P(c=1|t=1)}$, i.e. non-tested people are twice as likely infected than tested ones. Formally, this case is obtained with $\underline{\kappa}=1$ and $\overline{\kappa}=2$. For Graz this additional assumption gives prevalence bounds [$0.49\%$, $3.01\%$]. Here, and for all the other counties as well, the upper bound is reduced by a large amount relative to the Negative Selection scenario. Thus, this assumption of restricted negative selection has quite some identification power, although ``twice as likely'' does not appear to be unrealistic.
	\item[Small Ambigous Selection.] Finally, consider a case where no knowledge about the direction of selection is warranted, but there is evidence for small selection, meaning ${P(c=1|t=0)}$ is close to ${P(c=1|t=1)}$. In particular, here the assumption is that non-tested people have an infection probability between $95\%$ and $105\%$ of the tested people's infection probability (i.e. $\underline{\kappa}=0.95$ and $\overline{\kappa}=1.05$). In this scenario, bounds are pretty tight and for Graz they are [$0.47\%$, $1.75\%$]. Since there is the possibility of positive selection, the lower bound here is lower than in the previous cases.
\end{description}

\newpage
{\footnotesize
		\begin{longtable}{@{}ccccc@{}} \caption{County-level prevalence with no or very strong assumptions} \label{tab:1}\\
		\toprule
		County & Positive Tests & Participation &  No Assumption &  No Selection      \\
		\midrule                                
Amstetten, Scheibbs, Waidhofen adY & 0.21\% & 30.98\% & [0.00\%, 69.14\%] & [0.00\%, 0.40\%] \\[0.15ex]
Bruck an der Leitha & 0.10\% & 39.21\% & [0.00\%, 60.87\%] & [0.00\%, 0.20\%] \\[0.15ex]               
Baden & 0.10\% & 31.65\% & [0.00\%, 68.41\%] & [0.00\%, 0.19\%] \\[0.15ex]                             
Gmuend & 0.29\% & 28.95\% & [0.00\%, 71.20\%] & [0.00\%, 0.54\%] \\[0.15ex]                            
Gaenserndorf & 0.11\% & 34.99\% & [0.00\%, 65.08\%] & [0.00\%, 0.21\%] \\[0.15ex]                      
Hollabrunn & 0.15\% & 37.91\% & [0.00\%, 62.20\%] & [0.00\%, 0.27\%] \\[0.15ex]                        
Horn & 0.05\% & 37.26\% & [0.00\%, 62.78\%] & [0.00\%, 0.09\%] \\[0.15ex]                              
Korneuburg & 0.15\% & 47.08\% & [0.00\%, 53.05\%] & [0.00\%, 0.27\%] \\[0.15ex]                        
Krems (Land) & 0.09\% & 35.67\% & [0.00\%, 64.39\%] & [0.00\%, 0.18\%] \\[0.15ex]                      
Krems (Stadt) & 0.12\% & 27.79\% & [0.00\%, 72.27\%] & [0.00\%, 0.23\%] \\[0.15ex]                     
Lilienfeld & 0.12\% & 33.07\% & [0.00\%, 67.00\%] & [0.00\%, 0.23\%] \\[0.15ex]                        
Moedling & 0.13\% & 37.98\% & [0.00\%, 62.11\%] & [0.00\%, 0.23\%] \\[0.15ex]                          
Melk & 0.20\% & 31.11\% & [0.00\%, 69.01\%] & [0.00\%, 0.38\%] \\[0.15ex]                              
Mistelbach & 0.13\% & 46.67\% & [0.00\%, 53.45\%] & [0.00\%, 0.25\%] \\[0.15ex]                        
Neunkirchen & 0.19\% & 29.14\% & [0.00\%, 70.97\%] & [0.00\%, 0.36\%] \\[0.15ex]                       
St. Poelten (Stadt) & 0.19\% & 22.46\% & [0.00\%, 77.62\%] & [0.00\%, 0.35\%] \\[0.15ex]               
St. Poelten (Land) & 0.16\% & 39.78\% & [0.00\%, 60.33\%] & [0.00\%, 0.30\%] \\[0.15ex]                
Tulln & 0.09\% & 30.23\% & [0.00\%, 69.83\%] & [0.00\%, 0.18\%] \\[0.15ex]                             
Wiener Neustadt (Stadt+Land) & 0.13\% & 29.43\% & [0.00\%, 70.64\%] & [0.00\%, 0.23\%] \\[0.15ex]      
Waidhofen an der Thaya & 0.13\% & 26.18\% & [0.00\%, 73.88\%] & [0.00\%, 0.23\%] \\[0.15ex]            
Zwettl & 0.17\% & 31.98\% & [0.00\%, 68.12\%] & [0.00\%, 0.32\%] \\[0.15ex]                            
Innsbruck-Stadt & 0.19\% & 29.21\% & [0.00\%, 70.89\%] & [0.00\%, 0.35\%] \\[0.15ex]                   
Imst & 0.32\% & 28.80\% & [0.00\%, 71.37\%] & [0.00\%, 0.60\%] \\[0.15ex]                              
Innsbruck-Land & 0.22\% & 34.75\% & [0.00\%, 65.39\%] & [0.00\%, 0.41\%] \\[0.15ex]                    
Kitzbuehel & 0.28\% & 31.42\% & [0.00\%, 68.74\%] & [0.00\%, 0.52\%] \\[0.15ex]                        
Kufstein & 0.30\% & 32.43\% & [0.00\%, 67.75\%] & [0.00\%, 0.56\%] \\[0.15ex]                          
Landeck & 0.36\% & 32.44\% & [0.00\%, 67.78\%] & [0.00\%, 0.68\%] \\[0.15ex]                           
Lienz & 0.37\% & 25.25\% & [0.00\%, 74.93\%] & [0.00\%, 0.69\%] \\[0.15ex]                             
Reutte & 0.18\% & 35.48\% & [0.00\%, 64.64\%] & [0.00\%, 0.34\%] \\[0.15ex]                            
Schwaz & 0.48\% & 28.26\% & [0.00\%, 71.99\%] & [0.00\%, 0.89\%] \\[0.15ex]                            
Bruck-Muerzzuschlag & 0.49\% & 22.98\% & [0.00\%, 77.22\%] & [0.02\%, 0.91\%] \\[0.15ex]               
Deutschlandsberg & 0.37\% & 21.13\% & [0.00\%, 79.01\%] & [0.00\%, 0.69\%] \\[0.15ex]                  
Graz-Stadt & 0.90\% & 20.65\% & [0.10\%, 79.70\%] & [0.49\%, 1.68\%] \\[0.15ex]                        
Graz-Umgebung & 0.19\% & 26.89\% & [0.00\%, 73.20\%] & [0.00\%, 0.36\%] \\[0.15ex]                     
Hartberg-Fuerstenfeld & 0.19\% & 19.97\% & [0.00\%, 80.10\%] & [0.00\%, 0.36\%] \\[0.15ex]             
Leibnitz & 0.19\% & 19.47\% & [0.00\%, 80.60\%] & [0.00\%, 0.36\%] \\[0.15ex]                          
Leoben & 0.19\% & 20.68\% & [0.00\%, 79.39\%] & [0.00\%, 0.35\%] \\[0.15ex]                            
Liezen & 0.32\% & 16.91\% & [0.00\%, 83.19\%] & [0.00\%, 0.60\%] \\[0.15ex]                            
Murau & 0.31\% & 18.28\% & [0.00\%, 81.83\%] & [0.00\%, 0.58\%] \\[0.15ex]                             
Murtal & 0.32\% & 18.42\% & [0.00\%, 81.69\%] & [0.00\%, 0.59\%] \\[0.15ex]                            
Suedoststeiermark & 0.21\% & 22.17\% & [0.00\%, 77.92\%] & [0.00\%, 0.39\%] \\[0.15ex]                 
Voitsberg & 0.16\% & 20.27\% & [0.00\%, 79.79\%] & [0.00\%, 0.30\%] \\[0.15ex]                         
Weiz & 0.16\% & 19.11\% & [0.00\%, 80.95\%] & [0.00\%, 0.30\%] \\[0.15ex]                                                    
Salzburg (Stadt) & 0.27\% & 19.08\% & [0.00\%, 81.02\%] & [0.00\%, 0.50\%] \\[0.15ex]                  
Hallein (Tennengau) & 0.73\% & 23.87\% & [0.07\%, 76.45\%] & [0.30\%, 1.36\%] \\[0.15ex]               
Salzburg-Umgebung (Flachgau) & 0.37\% & 29.24\% & [0.00\%, 70.96\%] & [0.00\%, 0.68\%] \\[0.15ex]      
St. Johann im Pongau (Pongau) & 0.42\% & 22.41\% & [0.00\%, 77.76\%] & [0.00\%, 0.78\%] \\[0.15ex]     
Tamsweg (Lungau) & 0.71\% & 24.17\% & [0.07\%, 76.15\%] & [0.28\%, 1.33\%] \\[0.15ex]                  
Zell am See (Pinzgau) & 0.50\% & 24.02\% & [0.01\%, 76.20\%] & [0.03\%, 0.94\%] \\[0.15ex]             
Linz (Stadt) & 0.34\% & 20.90\% & [0.00\%, 79.23\%] & [0.00\%, 0.63\%] \\[0.15ex]                      
Steyr (Stadt) & 0.33\% & 22.09\% & [0.00\%, 78.05\%] & [0.00\%, 0.61\%] \\[0.15ex]                     
Wels (Stadt) & 0.38\% & 14.79\% & [0.00\%, 85.32\%] & [0.00\%, 0.71\%] \\[0.15ex]                      
Braunau am Inn & 0.45\% & 18.39\% & [0.00\%, 81.76\%] & [0.00\%, 0.84\%] \\[0.15ex]                    
Eferding & 0.26\% & 26.31\% & [0.00\%, 73.82\%] & [0.00\%, 0.48\%] \\[0.15ex]                          
Freistadt & 0.30\% & 26.85\% & [0.00\%, 73.30\%] & [0.00\%, 0.56\%] \\[0.15ex]                         
Gmunden & 0.45\% & 24.10\% & [0.00\%, 76.10\%] & [0.00\%, 0.83\%] \\[0.15ex]                           
Grieskirchen & 0.35\% & 24.83\% & [0.00\%, 75.33\%] & [0.00\%, 0.65\%] \\[0.15ex]                      
Kirchdorf & 0.33\% & 25.15\% & [0.00\%, 75.01\%] & [0.00\%, 0.62\%] \\[0.15ex]                         
Linz-Land & 0.41\% & 24.74\% & [0.00\%, 75.45\%] & [0.00\%, 0.77\%] \\[0.15ex]                         
Perg & 0.25\% & 24.03\% & [0.00\%, 76.09\%] & [0.00\%, 0.47\%] \\[0.15ex]                              
Ried im Innkreis & 0.31\% & 20.35\% & [0.00\%, 79.77\%] & [0.00\%, 0.57\%] \\[0.15ex]                  
Rohrbach & 0.54\% & 27.41\% & [0.02\%, 72.87\%] & [0.08\%, 1.00\%] \\[0.15ex]                          
Schaerding & 0.86\% & 23.81\% & [0.10\%, 76.57\%] & [0.44\%, 1.60\%] \\[0.15ex]                        
Steyr-Land & 0.57\% & 20.14\% & [0.02\%, 80.08\%] & [0.12\%, 1.07\%] \\[0.15ex]                        
Urfahr-Umgebung & 0.27\% & 26.56\% & [0.00\%, 73.57\%] & [0.00\%, 0.51\%] \\[0.15ex]                   
Voecklabruck & 0.48\% & 25.37\% & [0.00\%, 74.86\%] & [0.01\%, 0.90\%] \\[0.15ex]                      
Wels-Land & 0.36\% & 21.45\% & [0.00\%, 78.69\%] & [0.00\%, 0.67\%] \\[0.15ex]                         
Eisenstadt (Stadt+Umgebung) / Rust & 0.12\% & 28.12\% & [0.00\%, 71.95\%] & [0.00\%, 0.22\%] \\[0.15ex]
Guessing & 0.27\% & 21.22\% & [0.00\%, 78.89\%] & [0.00\%, 0.50\%] \\[0.15ex]                          
Jennersdorf & 0.22\% & 24.50\% & [0.00\%, 75.60\%] & [0.00\%, 0.42\%] \\[0.15ex]                       
Mattersburg & 0.07\% & 25.01\% & [0.00\%, 75.03\%] & [0.00\%, 0.14\%] \\[0.15ex]                       
Neusiedl am See & 0.11\% & 29.95\% & [0.00\%, 70.12\%] & [0.00\%, 0.21\%] \\[0.15ex]                   
Oberpullendorf & 0.17\% & 28.77\% & [0.00\%, 71.32\%] & [0.00\%, 0.31\%] \\[0.15ex]                    
Oberwart & 0.26\% & 22.04\% & [0.00\%, 78.07\%] & [0.00\%, 0.49\%] \\[0.15ex]                          
Wien & 0.32\% & 13.10\% & [0.00\%, 86.98\%] & [0.00\%, 0.60\%] \\[0.15ex]                                   
		\bottomrule                                                                                                                                   
	\end{longtable}}

{\footnotesize
	\begin{longtable}{@{}cccc@{}} \caption{County-level prevalence with selection assumptions} \label{tab:2}\\
		\toprule
		County $\downarrow \; \diagdown \; (\underline{\kappa},\overline{\kappa})\rightarrow$ & $(1,\infty)$ & $(1,2)$ &  $(0.95,1.05)$       \\
		\midrule                                
Amstetten, Scheibbs, Waidhofen adY & [0.00\%, 69.14\%] & [0.00\%, 0.67\%] & [0.00\%, 0.41\%] \\[0.15ex]
Bruck an der Leitha & [0.00\%, 60.87\%] & [0.00\%, 0.31\%] & [0.00\%, 0.20\%] \\[0.15ex]               
Baden & [0.00\%, 68.41\%] & [0.00\%, 0.32\%] & [0.00\%, 0.19\%] \\[0.15ex]                             
Gmuend & [0.00\%, 71.20\%] & [0.00\%, 0.92\%] & [0.00\%, 0.56\%] \\[0.15ex]                            
Gaenserndorf & [0.00\%, 65.08\%] & [0.00\%, 0.35\%] & [0.00\%, 0.22\%] \\[0.15ex]                      
Hollabrunn & [0.00\%, 62.20\%] & [0.00\%, 0.44\%] & [0.00\%, 0.28\%] \\[0.15ex]                        
Horn & [0.00\%, 62.78\%] & [0.00\%, 0.14\%] & [0.00\%, 0.09\%] \\[0.15ex]                              
Korneuburg & [0.00\%, 53.05\%] & [0.00\%, 0.41\%] & [0.00\%, 0.28\%] \\[0.15ex]                        
Krems (Land) & [0.00\%, 64.39\%] & [0.00\%, 0.29\%] & [0.00\%, 0.18\%] \\[0.15ex]                      
Krems (Stadt) & [0.00\%, 72.27\%] & [0.00\%, 0.39\%] & [0.00\%, 0.23\%] \\[0.15ex]                     
Lilienfeld & [0.00\%, 67.00\%] & [0.00\%, 0.39\%] & [0.00\%, 0.24\%] \\[0.15ex]                        
Moedling & [0.00\%, 62.11\%] & [0.00\%, 0.38\%] & [0.00\%, 0.24\%] \\[0.15ex]                          
Melk & [0.00\%, 69.01\%] & [0.00\%, 0.63\%] & [0.00\%, 0.39\%] \\[0.15ex]                              
Mistelbach & [0.00\%, 53.45\%] & [0.00\%, 0.38\%] & [0.00\%, 0.25\%] \\[0.15ex]                        
Neunkirchen & [0.00\%, 70.97\%] & [0.00\%, 0.62\%] & [0.00\%, 0.37\%] \\[0.15ex]                       
St. Poelten (Stadt) & [0.00\%, 77.62\%] & [0.00\%, 0.62\%] & [0.00\%, 0.36\%] \\[0.15ex]               
St. Poelten (Land) & [0.00\%, 60.33\%] & [0.00\%, 0.47\%] & [0.00\%, 0.30\%] \\[0.15ex]                
Tulln & [0.00\%, 69.83\%] & [0.00\%, 0.30\%] & [0.00\%, 0.18\%] \\[0.15ex]                             
Wiener Neustadt (Stadt+Land) & [0.00\%, 70.64\%] & [0.00\%, 0.40\%] & [0.00\%, 0.24\%] \\[0.15ex]      
Waidhofen an der Thaya & [0.00\%, 73.88\%] & [0.00\%, 0.41\%] & [0.00\%, 0.24\%] \\[0.15ex]            
Zwettl & [0.00\%, 68.12\%] & [0.00\%, 0.54\%] & [0.00\%, 0.33\%] \\[0.15ex]                            
Innsbruck-Stadt & [0.00\%, 70.89\%] & [0.00\%, 0.59\%] & [0.00\%, 0.36\%] \\[0.15ex]                   
Imst & [0.00\%, 71.37\%] & [0.00\%, 1.02\%] & [0.00\%, 0.62\%] \\[0.15ex]                              
Innsbruck-Land & [0.00\%, 65.39\%] & [0.00\%, 0.68\%] & [0.00\%, 0.43\%] \\[0.15ex]                    
Kitzbuehel & [0.00\%, 68.74\%] & [0.00\%, 0.87\%] & [0.00\%, 0.54\%] \\[0.15ex]                        
Kufstein & [0.00\%, 67.75\%] & [0.00\%, 0.93\%] & [0.00\%, 0.58\%] \\[0.15ex]                          
Landeck & [0.00\%, 67.78\%] & [0.00\%, 1.13\%] & [0.00\%, 0.70\%] \\[0.15ex]                           
Lienz & [0.00\%, 74.93\%] & [0.00\%, 1.21\%] & [0.00\%, 0.72\%] \\[0.15ex]                             
Reutte & [0.00\%, 64.64\%] & [0.00\%, 0.56\%] & [0.00\%, 0.35\%] \\[0.15ex]                            
Schwaz & [0.00\%, 71.99\%] & [0.00\%, 1.52\%] & [0.00\%, 0.92\%] \\[0.15ex]                            
Bruck-Muerzzuschlag & [0.02\%, 77.22\%] & [0.02\%, 1.61\%] & [0.02\%, 0.94\%] \\[0.15ex]               
Deutschlandsberg & [0.00\%, 79.01\%] & [0.00\%, 1.23\%] & [0.00\%, 0.71\%] \\[0.15ex]                  
Graz-Stadt & [0.49\%, 79.70\%] & [0.49\%, 3.01\%] & [0.47\%, 1.75\%] \\[0.15ex]                        
Graz-Umgebung & [0.00\%, 73.20\%] & [0.00\%, 0.62\%] & [0.00\%, 0.37\%] \\[0.15ex]                     
Hartberg-Fuerstenfeld & [0.00\%, 80.10\%] & [0.00\%, 0.65\%] & [0.00\%, 0.37\%] \\[0.15ex]             
Leibnitz & [0.00\%, 80.60\%] & [0.00\%, 0.65\%] & [0.00\%, 0.37\%] \\[0.15ex]                          
Leoben & [0.00\%, 79.39\%] & [0.00\%, 0.62\%] & [0.00\%, 0.36\%] \\[0.15ex]                            
Liezen & [0.00\%, 83.19\%] & [0.00\%, 1.10\%] & [0.00\%, 0.62\%] \\[0.15ex]                            
Murau & [0.00\%, 81.83\%] & [0.00\%, 1.06\%] & [0.00\%, 0.61\%] \\[0.15ex]                             
Murtal & [0.00\%, 81.69\%] & [0.00\%, 1.08\%] & [0.00\%, 0.62\%] \\[0.15ex]                            
Suedoststeiermark & [0.00\%, 77.92\%] & [0.00\%, 0.70\%] & [0.00\%, 0.41\%] \\[0.15ex]                 
Voitsberg & [0.00\%, 79.79\%] & [0.00\%, 0.54\%] & [0.00\%, 0.31\%] \\[0.15ex]                         
Weiz & [0.00\%, 80.95\%] & [0.00\%, 0.54\%] & [0.00\%, 0.31\%] \\[0.15ex]                                                       
Salzburg (Stadt) & [0.00\%, 81.02\%] & [0.00\%, 0.90\%] & [0.00\%, 0.52\%] \\[0.15ex]                  
Hallein (Tennengau) & [0.30\%, 76.45\%] & [0.30\%, 2.40\%] & [0.28\%, 1.42\%] \\[0.15ex]               
Salzburg-Umgebung (Flachgau) & [0.00\%, 70.96\%] & [0.00\%, 1.17\%] & [0.00\%, 0.71\%] \\[0.15ex]      
St. Johann im Pongau (Pongau) & [0.00\%, 77.76\%] & [0.00\%, 1.38\%] & [0.00\%, 0.81\%] \\[0.15ex]     
Tamsweg (Lungau) & [0.28\%, 76.15\%] & [0.28\%, 2.34\%] & [0.27\%, 1.38\%] \\[0.15ex]                  
Zell am See (Pinzgau) & [0.03\%, 76.20\%] & [0.03\%, 1.65\%] & [0.03\%, 0.97\%] \\[0.15ex]             
Linz (Stadt) & [0.00\%, 79.23\%] & [0.00\%, 1.13\%] & [0.00\%, 0.65\%] \\[0.15ex]                      
Steyr (Stadt) & [0.00\%, 78.05\%] & [0.00\%, 1.08\%] & [0.00\%, 0.63\%] \\[0.15ex]                     
Wels (Stadt) & [0.00\%, 85.32\%] & [0.00\%, 1.32\%] & [0.00\%, 0.74\%] \\[0.15ex]                      
Braunau am Inn & [0.00\%, 81.76\%] & [0.00\%, 1.53\%] & [0.00\%, 0.88\%] \\[0.15ex]                    
Eferding & [0.00\%, 73.82\%] & [0.00\%, 0.83\%] & [0.00\%, 0.50\%] \\[0.15ex]                          
Freistadt & [0.00\%, 73.30\%] & [0.00\%, 0.96\%] & [0.00\%, 0.58\%] \\[0.15ex]                         
Gmunden & [0.00\%, 76.10\%] & [0.00\%, 1.46\%] & [0.00\%, 0.86\%] \\[0.15ex]                           
Grieskirchen & [0.00\%, 75.33\%] & [0.00\%, 1.14\%] & [0.00\%, 0.68\%] \\[0.15ex]                      
Kirchdorf & [0.00\%, 75.01\%] & [0.00\%, 1.09\%] & [0.00\%, 0.65\%] \\[0.15ex]                         
Linz-Land & [0.00\%, 75.45\%] & [0.00\%, 1.35\%] & [0.00\%, 0.80\%] \\[0.15ex]                         
Perg & [0.00\%, 76.09\%] & [0.00\%, 0.83\%] & [0.00\%, 0.49\%] \\[0.15ex]                              
Ried im Innkreis & [0.00\%, 79.77\%] & [0.00\%, 1.02\%] & [0.00\%, 0.59\%] \\[0.15ex]                  
Rohrbach & [0.08\%, 72.87\%] & [0.08\%, 1.73\%] & [0.07\%, 1.04\%] \\[0.15ex]                          
Schaerding & [0.44\%, 76.57\%] & [0.44\%, 2.81\%] & [0.42\%, 1.66\%] \\[0.15ex]                        
Steyr-Land & [0.12\%, 80.08\%] & [0.12\%, 1.93\%] & [0.11\%, 1.12\%] \\[0.15ex]                        
Urfahr-Umgebung & [0.00\%, 73.57\%] & [0.00\%, 0.88\%] & [0.00\%, 0.52\%] \\[0.15ex]                   
Voecklabruck & [0.01\%, 74.86\%] & [0.01\%, 1.57\%] & [0.01\%, 0.93\%] \\[0.15ex]                      
Wels-Land & [0.00\%, 78.69\%] & [0.00\%, 1.20\%] & [0.00\%, 0.70\%] \\[0.15ex]                         
Eisenstadt (Stadt+Umgebung) / Rust & [0.00\%, 71.95\%] & [0.00\%, 0.38\%] & [0.00\%, 0.23\%] \\[0.15ex]
Guessing & [0.00\%, 78.89\%] & [0.00\%, 0.90\%] & [0.00\%, 0.52\%] \\[0.15ex]                          
Jennersdorf & [0.00\%, 75.60\%] & [0.00\%, 0.73\%] & [0.00\%, 0.43\%] \\[0.15ex]                       
Mattersburg & [0.00\%, 75.03\%] & [0.00\%, 0.24\%] & [0.00\%, 0.14\%] \\[0.15ex]                       
Neusiedl am See & [0.00\%, 70.12\%] & [0.00\%, 0.36\%] & [0.00\%, 0.22\%] \\[0.15ex]                   
Oberpullendorf & [0.00\%, 71.32\%] & [0.00\%, 0.53\%] & [0.00\%, 0.32\%] \\[0.15ex]                    
Oberwart & [0.00\%, 78.07\%] & [0.00\%, 0.88\%] & [0.00\%, 0.51\%] \\[0.15ex]                          
Wien & [0.00\%, 86.98\%] & [0.00\%, 1.11\%] & [0.00\%, 0.62\%] \\[0.15ex]                                   
		\bottomrule                                                                                                                                   
\end{longtable}}

\bibliographystyle{ecta} 
\bibliography{prevalance_aut}


\end{document}